\documentclass[a4paper]{llncs}
\usepackage[utf8]{inputenc}

\usepackage{ltl}
\usepackage{graphicx}
\usepackage{subfig}
\usepackage{amsmath, amssymb}
\usepackage[ruled,vlined,linesnumbered,commentsnumbered]{algorithm2e}
\SetAlFnt{\large}

\usepackage{pslatex}

\usepackage{tikz}
\usetikzlibrary{patterns}
\usetikzlibrary{arrows,shapes,automata,backgrounds,petri}
\tikzstyle{vertex}=[circle,minimum size=8pt,inner sep=0pt]
\tikzstyle{vEllipse}=[ellipse,minimum height=15pt,minimum width=30pt,
  inner sep=0pt,draw]
\tikzstyle{edge} = [draw,thick,->]
\tikzstyle{selected edge} = [draw,line width=5pt,line join=round,-,red!80]
\pgfdeclarelayer{background}
\pgfsetlayers{background,main}
\tikzstyle{sec vertex}=[circle,draw,minimum size=20pt,inner sep=0pt]
\tikzstyle{thd vertex}=[circle,draw,minimum size=8pt,inner sep=0pt]

\newtheorem{thm}{Theorem}

\newtheorem{define}[thm]{Definition}

\title{Analysing Sanity of Requirements for Avionics Systems\thanks{ The
    research leading to these results has received funding from the European
    Union’s Seventh Framework Program (FP7/2007-2013) for CRYSTAL – Critical
    System Engineering Acceleration Joint Undertaking under grant agreement
    No. 332830 and from specific national programs and/or funding authorities.}\\(Preliminary Version)}

\author{Ji{\v r}{\'i} Barnat$^1$, Petr Bauch$^1$, Nikola Bene\v{s}$^1$, Lubo{\v
    s} Brim$^1$,\\ Jan Beran$^2$, and Tomáš Kratochv\'{i}la$^2$}

\institute{
  \begin{center}
    \begin{tabular}{cc}
      \begin{minipage}{0.45\linewidth}
        \begin{tabbing}
          1) \=Faculty of Informatics\\
          \>Masaryk University\\
          \>Brno, Czech Republic\\
        \end{tabbing}
      \end{minipage}
      & \quad
      \begin{minipage}{0.45\linewidth}
        \begin{tabbing}
          2) \=Honeywell International\\
          \>Aerospace, Advanced Technology Europe\\
          \>Brno, Czech Republic\\
        \end{tabbing}
      \end{minipage}
    \end{tabular}
  \end{center}
}

\begin{document}

\maketitle

\begin{abstract}
  In the last decade it became a common practice to formalise software
  requirements to improve the clarity of users' expectations. In this work we
  build on the fact that functional requirements can be expressed in temporal
  logic and we propose new sanity checking techniques that automatically detect
  flaws and suggest improvements of given requirements. Specifically, we
  describe and experimentally evaluate approaches to consistency and redundancy
  checking that identify all inconsistencies and pinpoint their exact source
  (the smallest inconsistent set). We further report on the experience obtained
  from employing the consistency and redundancy checking in an industrial
  environment. To complete the sanity checking we also describe a semi-automatic
  completeness evaluation that can assess the coverage of user requirements and
  suggest missing properties the user might have wanted to formulate. The
  usefulness of our completeness evaluation is demonstrated in a case study of
  an aeroplane control system.
\end{abstract}

\section{Introduction}

The earliest stages of software development entail among others the activity of
user requirements elicitation. The importance of clear specification of the
requirements in the contract-based development process is apparent from the
necessity of final-product compliance verification. Yet the specification itself
is rarely described formally. Nevertheless, the formal description is an
essential requirement for any kind of comprehensive verification. Recently,
there have been tendencies to use the mathematical language of temporal logics,
e.g. the \emph{Linear Temporal Logic} (LTL), to specify functional system
requirements. Restating requirements in a rigorous, formal way enables the
requirement engineers to scrutinise their insight into the problem and allows
for a considerably more thorough analysis of the final requirement
documents~\cite{HJC08}.

Later in the development, when the requirements are given and a model is
designed, the formal verification tools can provide a proof of correctness of
the system being developed with respect to formally written requirements. The
model of the system or even the system itself can be checked using model
checking~\cite{BBC10,CCG02} or theorem proving~\cite{BFG01} tools. If there are
some requirements the system does not meet, the cause has to be found and the
development reverted. The longer it takes to discover an error in the
development, the more expensive the error is to mend. Consequently, errors made
during requirements specification are among the most expensive ones in the whole
development.

Model checking is particularly efficient in finding bugs in the design, however,
it exhibits some shortcomings when it is applied to requirement analysis. In
particular, if the system satisfies the formula then the model checking
procedure only provides an affirmative answer and does not elaborate for the
reason of the satisfaction. It could be the case that, e.g. the specified
formula is a tautology, hence, it is satisfied for any system. Tautologies can
be obvious, such as $p\vee \neg p$ or subtle, such as $p\LTLu\neg(r\wedge q)$,
where $r$ stands for ``\emph{The value of variable x is smaller than 5}'' and
$q$ stands for ``\emph{The value of variable x is greater than 10}''. To
mitigate the situation a subsidiary approach, the sanity checking, was proposed
to check vacuity and coverage of requirements~\cite{Kup06}. Yet the existence of
a model is still a prerequisite which postpones the verification until later
phases in the development cycle.

The primary interest of this paper is to re-state and extend techniques of our
own~\cite{BBB12} that allow the developers of computer systems to check sanity
of their requirements when it matters most, i.e. during the requirements
stage. Sanity checking commonly consists of three parts: consistency checking,
redundancy checking, and checking completeness of requirements. A consistent set
of requirements is one that can be implemented in a single system. A vacuous
requirement is satisfied trivially by a given system and redundancy is the
equivalent of vacuity for the model-free case. Finally, a complete set of
requirements extends to all sensible behaviours of a system. All three notions
will be defined properly in Section~\ref{sec:pre}.

In our previous work we redefined the notion of sanity checking of requirements
written as LTL formulae and described its implementation and evaluation. The
novelty of our approach we recapitulate and extend in this paper lies in the
idea of liberating sanity checking from the necessity of having a model of the
system under development. Our approach to consistency and vacuity checking
presented before identified all inconsistent and vacuous subsets of the input
set of requirements. This considerably simplified the work of requirement
engineers because it allowed to pinpoint all the sources of inconsistencies. As
for completeness checking, we proposed a new behaviour-based coverage
metric. Assuming that the user specifies what behaviour of the system is
sensible, our coverage metric calculated what portion of this behaviour is
described by the requirements specifying the system itself. The method further
suggested new requirements to the user that would have improved the coverage and
thus ensured more accurate description of users' expectations.

\paragraph{Contribution}

The novelty of our consistency (redundancy) checking is that they produce all
inconsistent (redundant) sets instead of a yes/no answer and their efficiency is
demonstrated in an experimental evaluation. Finally, the completeness checking
presents a new behaviour-based coverage and suggests formulae that would improve
the coverage and, consequently, the rigour of the final requirements.

In this paper we extend~\cite{BBB12} with a modified workflow that allows
requirement engineers to explicitly specify universal ($\forall$) and
existential ($\exists$) nature of individual requirements allowing thus to more
naturally express some of the system properties. We also evaluate the robustness
and sensitivity of our consistency and completeness checking procedures with
respect to the use of different LTL to BA translators. Furthermore, we adopt
vacuity checking for individual LTL formulae into the framework and add proper
references to relevant previous work. Finally, we report on experience we
gathered in cooperation with Honeywell International when applying our
methodology on some real-life industrial use cases.

\subsection{Related Work}

The use of model checking with properties (specified in CTL) derived from
real-life avionics software specifications was successfully demonstrated
in~\cite{CAB98}. This paper intends to present a preliminary to such a
use of a model checking tool, because there the authors presupposed sanity of
their formulae. The idea of using coverage as a metric for completeness can be
traced back to software testing, where it is possible to use LTL requirements as
one of the coverage metrics~\cite{WRH06,RWH07}.

Model-based sanity checking was studied thoroughly and using various approaches,
but it is intrinsically different from model-free checking presented in this
paper. Completeness is measured using metrics based on the state space coverage
of the underlying model~\cite{CKK01,CKV01}. Vacuity of temporal formulae was
identified as a problem related to model checking and solutions were proposed
in~\cite{KV03} and in~\cite{Kup06}, again expecting existence of a model.

Checking consistency (or satisfiability) of temporal formulae is a well
understood problem solved using various techniques in many papers (most recently
using SAT-based approach in~\cite{RDB05} or in~\cite{RV07} where it was used as
a comparison between different LTL to B{\"u}chi translation techniques). The
classical problem is formulated as to decide whether a set of formulae is
internally consistent. In this paper, however, a more elaborate answer is
sought: specifically which of the formulae cause the inconsistency. The approach
is then extended to vacuity which is rarely used in model-free sanity checking.

We have adopted the idea of searching for the smallest inconsistent subset of a
set of requirements from the SAT community, where finding the minimal
unsatisfiable core (a set of propositional clauses) is an important feature of
SAT solvers, see e.g.~\cite{LS04}. The authors of~\cite{LS08} extend the
approach to finding all minimal unsatisfiable subsets, again in the context of
propositional logic. Their exclusion of checking those subsets containing known
unsatisfiable subsets was extended by our algorithm with a dual heuristic: exclude
also subsets with known satisfiable supersets.

A problem related to consistency is that of
\emph{realisability}~\cite{ALW89}. Realisability generalises consistency by
allowing some of the atomic propositions to be controlled by the environment. A
set of system requirements is then realisable if there is a strategy which can
react to arbitrary environment choices while satisfying the requirements. If no
propositions are controlled by the environment then the notion of realisability
becomes equivalent to consistency. A number of tools implement realisability,
RATSY~\cite{BCG10} for example, but again limit their responses to stating
whether the requirements are realisable as a whole. Thus, even though the theory
of minimal unrealisable cores had been known for some time~\cite{CRS08,KHB09},
searching for smallest inconsistent subsets is a feature unique for our tool. We
would also like to point out the theory described in~\cite{Sch12b} which extends
the unrealisability cores to search even within individual formulae, providing
even finer information on the quality of requirements.

Completeness of formal requirements is not as well-established and its
definition often differs. The most thorough research in algorithmic evaluation
of completeness was conducted in~\cite{HL95,Lev00,MTH03}. Authors of those
papers use RSML (Requirements State Machine Language) to specify requirements
which they translate (in the last paper) to CTL and to a theorem prover language
to discover inconsistencies and detect incompleteness. Their notion of
completeness is based on verifying that for every input value there is a
reaction described in the specification. This paper presents completeness as
considering all behaviours described as sensible (and either refuting or
requiring them). Finally, a novel semi-formal methodology is proposed in this
paper, that recommends new requirements to the user, that have the potential to
improve completeness of the input specification.

Another approach at defining and checking completeness was pursued
in~\cite{KGG99} in terms of evaluating the extend to which an implementation
describes the intended specification. Their evaluation is based on the
simulation preorder and searching for unimplemented transition/state
evidence. In order to create a metric that would drive the engine for generating
new requirements we have build an evaluation algorithm that enumerates the
similarity of two paths. Our notion of similarity is to certain extend based on
unimplemented transitions, but we carry the enumeration from two transition up
to the whole automata, thus finally obtaining our completeness metric. The
authors of~\cite{BB09} elaborate a method for measuring the quality of a set of
requirements by detecting the presence of such a subsets, that together require
behaviour that could not be detected by individual requirements. Similarly as
in~\cite{HL95}, the interesting behaviour is described in terms of reactions to
input values observed at discrete time instances, yet the properties are limited
to those of the form $\forall t, A_t(t)\implies Z_t(t)$.

Finally, the notion of coverage is well-established in the area of software
testing. The similarity with our work and that described for example
in~\cite{RLS03,FG03} lies in that both approaches attempt to generate a description
of sets of behaviours that remain uncovered. The dissimilarity lies in that we
describe these uncovered sets in term of software requirements (or LTL
formulae), whereas automated testing coverage produces collections of test,
i.e. individual runs of a particular system.

\section{Preliminaries}\label{sec:pre}

This section serves as a motivation for and a reminder of the model checking
process and its connection to sanity checking. A knowledgeable reader might find
it slow-paced and cursory, but its primary function is to justify the use of
formal specifications and, as such, requires more compliant approach.

\subsection{LTL Model Checking}

\begin{define} \label{def:ltl}
  Let $AP$ be the set of \emph{atomic propositions}. Then the recursive
  definition below specifies all well-formed \emph{LTL formulae} over $AP$,
  where $p\in AP$:
  \begin{center}
    $\varphi::=p\mid\neg\varphi\mid\varphi\wedge\varphi\mid \LTLx\varphi\mid
    \varphi\LTLu\varphi$
  \end{center}
\end{define}

\begin{example} \label{exa:ltl}
  There are some well-established syntactic simplifications of the LTL language,
  e.g. $\mathit{false}:=p\wedge\neg p$, $\mathit{true}:=\neg \mathit{false}$,
  $\varphi\Rightarrow\psi:=\neg(\varphi\wedge \neg\psi)$, $\LTLf\varphi:=$
  $\mathit{true}\LTLu\varphi$, $\LTLg\varphi:=\neg
  (\LTLf\neg\varphi)$. Assuming that $AP=\{\alpha:=(c=5),\beta:=(a\neq b)\}$,
  these are examples of well-formed LTL formulae: $\LTLg\beta,
  \alpha\LTLu\neg\beta$.\hfill$\triangle$
\end{example}

In classical model checking one usually verifies that a model of the system in
question satisfies the given set of LTL-specified requirements. That is not
possible in the context of this paper because in the requirements stage there is
no model to work with. Nevertheless, to better understand the background of LTL
model checking let us assume that the system is modelled as a \emph{Labelled
  Transition System} (LTS).

\begin{define} \label{def:lts}
  Let $\Sigma$ be a set of state labels (it will mostly hold that
  $\Sigma=AP$). Then an \emph{LTS} $M=(S, \rightarrow,\nu,S_0)$ is a tuple,
  where: $S$ is a set of states, $\rightarrow\subseteq S\times S$ is a
  transition relation, $\nu:S\rightarrow 2^{\Sigma}$ is a valuation function and
  $S_0\subseteq S$ is a~nonempty set of initial states. A function
  $r:\mathbb{N}\rightarrow S$ is an \emph{infinite run} over the states of $M$
  if $r(0)\in S_0, \forall i: r(i)\rightarrow r(i+1)$. The \emph{trace} or
  \emph{word} of a run is a function $w:\mathbb{N}\rightarrow 2^{\Sigma}$, where
  $w(i)=\nu(r(i))$.
\end{define}

An LTL formula states a property pertaining to an infinite trace (a trace that
does not have to be associated with a run). Assuming an LTS to be a model of a
computer program then a trace represents one specific execution of the
program. Also the infiniteness of the executions is not necessarily an error --
programs such as operating systems or controlling protocols are not supposed to
terminate.

\begin{define} \label{def:sem}
  Let $w$ be an infinite word and let $\varphi$ be an LTL formula over
  $\Sigma$. Then it is possible to decide if $w$ satisfies $\varphi$, $w\models
  \varphi$, based on the following rules:
  \[
  \begin{array}{lcl}
    w\models p &\ \mathit{iff} &\ p\in w(0),\\

    w\models \neg\varphi &\ \mathit{iff} &\ w\not\models\varphi,\\

    w\models \varphi_1\wedge\varphi_2 &\ \mathit{iff} &\ w\models \varphi_1
    \ \mathit{and}\ w\models \varphi_2,\\

    w\models \LTLx\varphi &\ \mathit{iff} &\ w_1\models\varphi,\\

    w\models \varphi_1\LTLu\varphi_2 &\ \mathit{iff} &\ \exists i.\forall j<i:
    w_j\models\varphi_1, w_i\models\varphi_2,
  \end{array}
  \]
  where $w(i)$ is the $i$-th letter of $w$ and $w_i$ is the $i$-th suffix of $w$.
\end{define}

\begin{example} \label{exa:ltlProp}
  Figure~\ref{fig:petProt} contains LTS for a process engaged in Peterson's
  mutual exclusion protocol. The protocol can control access to the critical
  section (state $CS$) for arbitrarily many processes that communicate using
  global variables to determine which process will be granted access next. The
  two LTL formulae verify the liveness property of the protocol:\\ 1: If a
  process \emph{waits} for an access to the critical section it will eventually
  get there.\\
  2: A process outside the critical section will eventually get inside,
  and this hold in any state of the system.\hfill$\triangle$
\end{example}

\begin{figure}[t]
  \centering
  \begin{tikzpicture}[scale=1.5,auto,swap]
  \clip (-0.4,-0.4) rectangle (6.6,1.2);
  \foreach \pos/\name in {{(0,0)/CS}, {(0,1)/NCS}, {(1.5,0.5)/wait}, {(3,0)/q3},
    {(3,1)/q2}} {
    \node[vEllipse] (\name) at \pos {\name};
  }
  \begin{scope}[>=angle 90]
    \foreach \from/\to in {CS/NCS, NCS/wait, wait/CS, wait/q2, q2/q3, q3/wait} {
      \path[edge] (\from) -- (\to);
    }
    \path[edge] (q3) to [out=45,in=325,loop] (q3);
  \end{scope}

  \node[vertex] (A) at (5.5,0.5)
    {$\begin{array}{l}
        1:\ \LTLg(wait\Rightarrow \LTLf CS)\\
        2:\ \LTLg(\neg CS\Rightarrow \LTLf CS)
      \end{array}$};
\end{tikzpicture}
  \caption{LTS for Peterson's mutual exclusion protocol (with only one process)
    and two liveness LTL formulae.}
  \label{fig:petProt}
\end{figure}
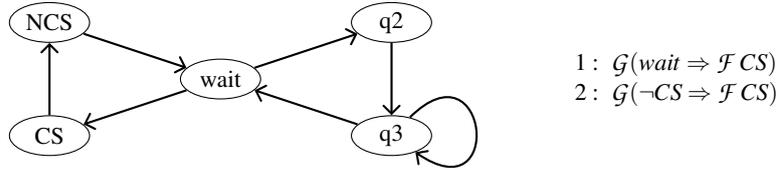

Traditionally, a system as a whole is considered to satisfy an LTL formula if
all its executions (all infinite words over the states of its LTS) do. However,
we might occasionally be also interested in the question whether at least one of
a~system's executions satisfies the given LTL formula. We thus suggest explicit
path quantification using the $\forall$ and $\exists$ quantifiers as follows:
For a system model $M$ and a formula $\varphi$, we write $M \models
\forall\varphi$ if all executions of $M$ satisfy $\varphi$, and $M \models
\exists\varphi$ if there is at least one execution of $M$ satisfying $\varphi$. We
call the former kind of path-quantified formulae \emph{universal} and the latter
kind \emph{existential}.

There are several notes to be made. First, we do not allow for nesting of
path-quantified formulae as we want to avoid the complexity of dealing with full
CTL$^*$. We only occasionally use the negation operator with the meaning $\neg
\forall \varphi \equiv \exists \neg \varphi$ and vice versa. Second, one of the
advantages of making the hidden universal quantification of LTL more explicit is
to reduce the confusion over the fact that LTL seemingly does not follow the law
of excluded middle: We can have a~system model that satisfies neither $\varphi$
nor $\neg\varphi$. Third, using path-quantified LTL formulae is really just
a~convenience, as it does not make the verification task any harder. Indeed, to
verify whether a~system model satisfies $\exists\varphi$ is the same as to
verify whether the system model satisfies $\forall\neg\varphi$ and invert the
answer. Efficient verification of that satisfaction, however, requires a more
systematic approach than enumeration of all executions. An example of a
successful approach is the enumerative approach using \emph{B{\"u}chi automata}.

\begin{define} \label{def:ba}
  A \emph{B{\"u}chi automaton} is a pair $A=(M,F)$, where $M$ is an LTS and
  $F\subseteq S$. An automaton $A$ accepts an infinite word $w$ ($w\in L(A)$) if
  there exists a run $r$ for $w$ in $M$ and there is a state from $F$ that
  appears infinitely often on $r$, i.e. $\forall i\exists j>i: r(j)\in F$.
\end{define}

Arbitrary LTL formula $\varphi$ can be transformed into a B{\"u}chi automaton
$A_{\varphi}$ such that $w\models\varphi\Leftrightarrow w\in
L(A_{\varphi})$. Also checking that every execution satisfies $\varphi$ is
equivalent to checking that no execution satisfies $\neg\varphi$. It only
remains to combine the LTS model of the given system $M$ with $A_{\neg\varphi}$
in such a way that the resulting automaton will accept exactly those words of
$M$ that violate $\varphi$. Finally, deciding existence of such a word -- and by
extension verifying correctness of the system -- has been shown equivalent to
finding accepting cycle in a graph.

In the following sections we shall usually deal with the standard approach to
LTL first, i.e.~consider implicit universal path quantification of the
formulae. We then follow with the extension to existential path-quantified
formulae.

\subsection{Model-based Sanity Checking}

As described in the introduction the model checking procedure is not designed to
decide why was a certain property satisfied in a given system. That is a
problem, however, because the reasons for satisfaction might be the wrong
ones. If for example the system is modelled erroneously or the formula is not
appropriate for the system then it is still possible to receive a positive
answer.

These kinds of inconsistencies between the model and the formula are detected
using sanity checking techniques, namely \emph{vacuity} and \emph{coverage}. In
this paper they will be described for comparison with their model-free versions;
interested reader should consult for example~\cite{Kup06} for more details.

Let $K$ be an LTS, $\varphi$ a formula and $\psi$ its subformula. Let further
$\varphi[\psi'/\psi]$ be the modified formula $\varphi$ in which its subformula
$\psi$ is substituted by $\psi'$. Then $\psi$ \emph{does not affect the truth
  value of} $\varphi$ \emph{in} $K$ if the following property holds: $K$
satisfies $\varphi[\psi'/\psi]$ for every formula $\psi'$ iff $K$ satisfies
$\varphi$. Then a system $K$ \emph{satisfies a formula} $\varphi$
\emph{vacuously} if $K$ satisfies $\varphi$ and there is a subformula $\psi$ of
$\varphi$ such that $\psi$ does not affect $\varphi$ in $K$.

A state $s$ of an LTS $K$ is $q$-\emph{covered by} $\varphi$, for a formula
$\varphi$ and an atomic proposition $q$, if $K$ satisfies $\varphi$ but
$\tilde{K}_{s,q}$ does not satisfy $\varphi$. There $\tilde{K}_{s,q}$ stands for
an LTS equivalent to $K$ except the valuation of $q$ in the state $s$ is
flipped.

If one could extract the underlying ideas of vacuity and coverage and show that
they do not necessarily require the existence of a model, it would be possible
to extrapolate these notions into the model-free environment. Vacuity states
that the satisfaction of a formula is given extrinsically (by \emph{an}
environment) and is not related to the formula itself. Thus we may consider the
remaining formulae from the set of requirements to constitute the
environment. Coverage, on the other hand, attempts to capture the amount of
system behaviour that is described by the formulae. Again, we can use the
remaining formulae to form the system, but in order to establish a coverage
metric capable of differentiating various formulae, we will need to extend the
notion from state-coverage to path-coverage.

Although these notions of sanity are dependent on the system that is being
verified, we might still utilise at least the notion of vacuity in a~model-less
setting, using a~concept of \emph{vacuity witnesses}. A~vacuity witness to
a~given formula $\varphi$ is a~formula $\psi$ with the following property: If
a~system $K$ satisfies $\psi$ then it satisfies $\varphi$ vacuously. A~standard
example is the request-response formula $\LTLg( req \Rightarrow \LTLf resp )$
``every request is followed by a response'', whose two vacuity witnesses are
$\LTLg( \neg req )$ ``no request ever happens'' and $\LTLg\LTLf( resp )$
``responses happen infinitely often''. The vacuity witnesses for a~given formula
can be automatically generated, see~\cite{Beer2001}. Moreover, as it is desired
for the vacuity witnesses not to hold, we may include this requirement in our
original set of requirements with the help of path-quantified LTL formulae. For
example, if our set of requirements contains the formula $\forall\LTLg( req
\Rightarrow \LTLf resp )$, we add the formulae $\exists\LTLf( req )$ and
$\exists\LTLf\LTLg( \neg resp )$ to the set -- those are the negations of the
vacuity witnesses.

In the rest of this paper, we translate the concepts of model-based sanity into
model-free environment. However, to avoid confusion, we use the word
\emph{redundancy} instead of model-less vacuity.\footnote{This is in contrast to
  the original paper~\cite{BBB12} where \emph{redundancy} has been called
  \emph{vacuity} as the paper did not include vacuity witnesses.} We further
supplement these notions with consistency verification thus arriving at a method
that would aid the creation of a reasonable set of requirements, i.e. a set such
that it is reasonable to ask if a system satisfies it or not. If, however, the
intention is to build a system based on this set, then a further property of
such set would be desirable and that the set completely describes the system
under development. Consequently, a \emph{sane} set of requirements is
consistent, without redundancies, and complete in a sense we will formalise in
Section~\ref{sec:compl}.

\section{Checking Consistency and Redundancy}

As various studies concluded, undetected errors made early in the development
cycle are the most expensive to eradicate. Thus it is very important that the
outcome of the requirements stage -- a database of well-formed, traceable
requirements -- is what the customer intended and that nothing was omitted (not
even unintentionally). While a procedure that would ensure the two properties
cannot be automated, this paper proposes a methodology to check the sanity of
requirements. In the following the \emph{sanity checking} will be considered to
consist of 3 related tasks: \emph{consistency}, \emph{redundancy} and
\emph{completeness} checking. As consistency and redundancy are more closely
related, we focus on them in this section and dedicate the following section
solely to completeness.

\begin{define} \label{def:consistency}
  A set $\Gamma$ of LTL formulae over $AP$ is \emph{consistent} if $\exists w\in
  AP^{\omega}:w$ satisfies $\bigwedge\Gamma$. Checking consistency of a set
  $\Gamma$ entails finding all minimal inconsistent subsets of $\Gamma$. A
  formula $\varphi$ is \emph{redundant} with respect to a set of formulae
  $\Gamma$ if $\bigwedge\Gamma\Rightarrow\varphi$. To check redundancy of a set
  $\Gamma$ entails finding all pairs of
  $\langle\varphi\in\Gamma,\Phi\subseteq\Gamma\rangle$ such that $\Phi$ is
  consistent and $\Phi\Rightarrow\varphi$ (and for no $\Phi' \subseteq \Phi$
  does it hold that $\Phi'\Rightarrow\varphi$).
\end{define}

The existence of the appropriate $w$ can be tested by constructing
$A_{\bigwedge\Gamma}$ and checking that $L(A_{\bigwedge\Gamma})$ is
non-empty. The procedure is effectively equivalent to model checking where the
model is a clique over the graph with one vertex for every element of $2^{AP}$
(allowing every possible behaviour).

This approach to consistency and redundancy is especially efficient if a
non-trivial set of requirements needs to be processed and standard sanity
checking would only reveal if there is an inconsistency (or redundancy) but
would not be able to locate the source. Furthermore, dealing with larger sets of
requirements entails the possibility that there will be several inconsistent
subsets or that a formula is redundant due to multiple small subsets. Each of
these conflicting subsets needs to be considered separately which can be
facilitated using the methodology proposed in this paper.

\begin{example} \label{exa:sanEx}
  Let us assume that there are five requirements formalised as LTL formulae over
  a set of atomic propositions (two-valued signals) $\{p,q,a\}$. They are
  \begin{enumerate}
  \item Eventually, the presence of signal $p$ will entail its continued
    presence until the signal $q$ is observed.\hfill$\varphi_1=\LTLf(p\Rightarrow
    p\LTLu q)$
  \item The signal $p$ will occur infinitely often.\hfill$\varphi_2=\LTLg\LTLf
    p$
  \item The signals $a$ and $p$ cannot occur at the same
    time. \hfill$\varphi_3=\LTLg\neg(a\wedge p)$
  \item Each occurrence of the signal $q$ requires that signal $a$ was observed
    in the previous time step. \hfill$\varphi_4=\LTLg(\LTLx
  q\Rightarrow a)$
  \item The signal $q$ will occur infinitely often.\hfill$\varphi_5=\LTLg\LTLf
    q$
  \end{enumerate}
  In this set the formula $\varphi_4$ is inconsistent due to the first 3
  formulae and the last formula is redundant with respect to (implied by) the
  first 2 formulae.\hfill$\triangle$
\end{example}

We now reformulate the notions of Definition~\ref{def:consistency} to extend
also to path-quantified LTL formulae. Again, we are interested in finding the
minimal inconsistent subsets and minimal redundancy witnesses.

\begin{define}
  A~set $\Gamma$ of path-quantified LTL formulae over $AP$ is \emph{consistent}
  if there is a~system model $M$ such that $M \models \varphi$ for all $\varphi
  \in \Gamma$. A~path-quantified formula $\varphi$ is \emph{redundant} with
  respect to a~set of formulae $\Gamma$ if every model $M$ satisfying all
  formulae of $\Gamma$ also satisfies $\varphi$.
\end{define}

We now show that checking consistency and redundancy in the path-quantified
setting can be easily reduced to the case in the standard setting. Let us
partition the set $\Gamma$ into existential formulae $\exists\Gamma$ and
universal formulae $\forall\Gamma$, i.e. $\exists\Gamma:= \Gamma \cap
\{\exists\varphi\}$, for $\varphi$ quantifier-free, and
$\forall\Gamma:=\Gamma\setminus \exists\Gamma$. Further, let $\Psi1\exists
:=\{\Phi\cup\{\varphi\}\mid \Phi\subseteq \forall\Gamma, \varphi \in
\exists\Gamma\}$.

\begin{theorem}
  In order to locate all smallest inconsistent subsets of $\Gamma$ is suffices
  to search among $\Psi1\exists$. Similarly, checking
  redundancy may correctly be limited to checking among pairs
  \[
  \begin{array}{ll}
    \langle\varphi\in\Gamma,\Phi\subseteq \forall\Gamma\rangle &
    \varphi\ \text{universal}\\
    \langle \forall\psi,\Phi\in \Psi1\exists\cup \forall\Gamma \rangle &
    \varphi\ \text{existential}, \varphi=\exists\psi.
  \end{array}
  \]
\end{theorem}

\begin{proof}
  Let us first consider consistency. Clearly, if $\Gamma$ only consists of
  universal formulae, we may simply ignore their quantifiers. In the case of
  $\Gamma$ containing existential formulae, we can make the following two
  observations:
  \begin{itemize}
  \item First observation: Two existential formulae are always consistent
    provided that they are both self-consistent, i.e.~none of them is equal to
    \emph{false}. This can be easily seen as we can always provide a~model with
    exactly two paths, each satisfying one of the formulae. This observations tells
    us that every minimal inconsistent subset contains at most one existential
    formula.
  \item Second observation: Let us write $\Gamma^\forall$ to denote the set
    $\Gamma$ in which all path quantifiers have been changed to $\forall$. If
    the set $\Gamma$ contains at most one existential formula then the following
    holds: $\Gamma$ is consistent iff $\Gamma^\forall$ is consistent. One
    direction is obvious, the other follows from the fact that we can always
    provide a~single-path model as a~proof of consistency of $\Gamma^\forall$.
  \end{itemize}
  The result of the two observations is that if we are careful to never consider
  sets with more than one existential formula, we may freely ignore the
  path-quantifiers.

  Let us now consider redundancy. It is clear that the following holds:
  $\varphi$ is redundant with respect to $\Gamma$ iff $\Gamma \cup \{ \neg
  \varphi \}$ is inconsistent. This means that redundancy checking is easily
  reduced to consistency checking. This fact holds in the path-quantified
  setting as well as in the standard one and will be used later in the
  implementation. Note that the first observation above has the following
  consequences: If $\varphi$ is a~universal formula, then it makes sense to
  consider $\Gamma$ with universal formulae only, as $\neg \varphi$ is
  existential. Furthermore, if $\varphi$ is an existential formula, it makes
  sense to consider $\Gamma$ containing at most one existential
  formula.\hfill\qed
\end{proof}

From the above theorem it follows that all formulae that require consideration
are universal, which is the implicit quantifier of quantifier-free formulae. In
other words the two statements above enable the same algorithms for checking
sanity, described in the next section, to be used in the path-quantified setting
as well. Yet the validity of the two statements requires a more detailed
argumentation. Before we come to the implementation, let us illustrate the
interplay between the concepts discussed so far (model-less consistency,
redundancy, and vacuity witnesses) on a more comprehensive example.

\begin{example}
Consider a~set of requirements containing just two items: 
\begin{enumerate}
\item If a signal $a$ is ever activated, it will remain active indefinitely,
  i.e. in every time step it holds that if $a$ is active now it will also be
  active in the next time step. Translating this natural language requirement to
  LTL we arrive at:\hfill $\varphi_1=\LTLg( a \Rightarrow \LTLx a )$
\item If a signal $a$ is ever activated, it will not be active in the next time
  step, i.e. a must never hold in two successive steps. After translation the
  appropriate LTL formula thus stands:\hfill $\varphi_2=\LTLg( a
  \Rightarrow \LTLx \neg a )$
\end{enumerate}
Note that although such requirements might seem inconsistent at first glance,
they are, in fact, consistent as any consistency checking algorithm could
demonstrate. What the following redundancy analysis is about to reveal is that,
while consistent, it is not reasonable to require both formulae to hold at the
same time.

Let us elaborate and formalise the reasoning step of redundancy checking. We
begin by constructing the vacuity witnesses according to the method
in~\cite{Beer2001}. The vacuity witnesses for $\varphi_1$ are: $\LTLg(\neg a)$,
$\LTLg\LTLx a$. The first witness is quite intuitive, if the signal $a$ is never
activated in a system then requiring $\varphi_1$ to hold is
unreasonable. Similarly, if at every time step we have that $a$ will be active
in the next step then clearly the presence of $a$ in this time step is
irrelevant. Exactly the same kind of mechanisable reasoning leads to the vacuity
witnesses for $\varphi_2$: $\LTLg(\neg a)$, $\LTLg\LTLx(\neg a)$. We thus create
the new requirements as described in the previous section. The set of
requirements now contains $\forall\varphi_1$, $\forall\varphi_2$, $\exists\LTLf
a$, $\exists\LTLf\LTLx(\neg a)$, $\exists\LTLf\LTLx a$. Take for example
$\exists\LTLf\LTLx(\neg a)$, resulting from the negated witness $\LTLg\LTLx
a$. Being existentially quantified, this requirement demands the existence of at
least one computation of the system, one on which a state whose successor does
not observe $a$ will eventually occur. Note that we have converted the implicit
formulae into universal ones and ignored the duplicity of the vacuity witnesses.

We now check the redundancy of the requirements and discover that the formula
$\exists\LTLf\LTLx a$ implies $\exists\LTLf a$, we thus drop $\exists\LTLf a$
from the set. The remaining four formulae are checked for consistency. We find
$\{\forall\varphi_1, \forall\varphi_2, \exists\LTLf\LTLx a\}$ as the only
minimal inconsistent subset. The conclusion is that although $\varphi_1$ and
$\varphi_2$ are consistent, the only models that satisfy both of them satisfy
them redundantly. The two formulae can thus be said to be redundantly
consistent.\hfill$\triangle$
\end{example}

\subsection{Implementation of Consistency Checking}

Let us henceforth denote one specific instance of consistency (or redundancy)
checking as a \emph{check}. For consistency and a set $\Gamma$ it means to check
that for some $\gamma\subseteq\Gamma$ is $\bigwedge\gamma$ satisfiable. For
redundancy it means for $\gamma\subseteq\Gamma$ and $\varphi\in\Gamma$ to check
that $\bigwedge\gamma\Rightarrow\varphi$ is satisfiable. In the worst case both
consistency and redundancy checking would require an exponential number of
checks. However, the proposed algorithm considers previous results and only
performs the checks that need to be tested.

Both consistency and redundancy checking use three data structures that facilitate
the computation. First, there is the queue of verification tasks called
$\mathit{Pool}$, then there are two sets, $\mathit{Con}$ and $\mathit{Incon}$,
which store the consistent and inconsistent combinations found so far. Finally,
each individual \emph{task} contains a set of integers (that uniquely identifies
formulae from $\Gamma$) and a \emph{flag} value (containing three bits for three
binary properties). First, whether the satisfaction check was already performed
or not. Second, if the combination is consistent. And the third bit specifies
the direction in subset relation (up or down in the Hasse diagram) in which the
algorithm will continue. The successors will be either subsets or supersets of
the current combination.

\begin{figure}[t]
\begin{minipage}{0.49\textwidth}
\begin{algorithm}[H]
  {\normalsize
  \input{consistency.tex}
  \caption{Consistency Check}
  \label{alg:con}}
\end{algorithm}
\end{minipage}
\vrule\hspace{0.01em}
\begin{minipage}{0.49\textwidth}
\begin{algorithm}[H]
  {\normalsize
  \input{successors.tex}
  \caption{\texttt{genSuccs}($\mathit{Task}$ \textsf{t})}
  \label{alg:succs}}
\end{algorithm}
\end{minipage}
\end{figure}

The idea behind consistency checking is very simple (listed as
Algorithm~\ref{alg:con}). The pool contains all the tasks to be performed and
these tasks are of two types: either to check consistency of the combination or
to generate successors. The symmetry of the solution allows for parallel
processing (multiple threads performing the Algorithm~\ref{alg:con} at the same
time) given that the data structures are properly protected from race
conditions. The pool needs to be initialised with all single element subsets of
$\Gamma$ and $\Gamma$ itself, thus in the subsequent iteration will be checked
the supersets of the former and subsets of the latter.

Algorithm~\ref{alg:succs} is called when the task \textsf{t} on the top of
$\mathit{Pool}$ is already checked. At this point either all subsets or all
supersets of \textsf{t} should be enqueued as tasks. But not all successors need
to be inspected, e.g. if \textsf{t} is consistent then also all its subsets will
be consistent -- that is clearly true and no subset of \textsf{t} needs to be
checked.

That observation is utilised again in Algorithm~\ref{alg:sups}. It does not
suffice to stop generating subsets and supersets when its immediate predecessors
are found consistent (inconsistent), because it can also happen that the
combination to be checked was formed in a different branch of the Hasse diagram
of the subset relation. In order to prevent redundant satisfiability checks two
sets are maintained $\mathit{Con}$ and $\mathit{Incon}$ (see how these are used
on line~\ref{algLine:sups1} of Algorithm~\ref{alg:sups}).

\begin{figure}[t]
\begin{minipage}{0.49\textwidth}

\begin{algorithm}[H]
  {\normalsize
  \input{superset.tex}
  \caption{\small\texttt{genSupsets}($\mathit{Task}$ \textsf{t})}
  \label{alg:sups}}
\end{algorithm}

\end{minipage}
\vrule\hspace{0.01em}
\begin{minipage}{0.49\textwidth}

\begin{algorithm}[H]
  {\normalsize
  \input{verification.tex}
  \caption{\small\texttt{verCons}(\textsf{t}=$\langle i_1,\ldots,i_j\rangle$)}
  \label{alg:ver}}
\end{algorithm}

\end{minipage}
\end{figure}

The actual consistency (and quite similarly also redundancy) checking is less
complicated (see Algorithm~\ref{alg:ver}) and may even be delegated to a third
party tool. First, the conjunction of formulae encoded in the task is
created. In our setting we then check the appropriate B{\"u}chi automaton
for the existence of an accepting cycle: using nested DFS~\cite{CVW92}. Hence
the Algorithm~\ref{alg:ver} can easily be substituted with
another method for checking consistency of a set of LTL formulae. A
realisability checking tool, such as RATSY, could also be applied: one only
needs to state that all signals are controlled by the system.

\subsubsection{Extension to Redundancy Checking}

The only difference when performing the redundancy checking is that the task
\textsf{t} consists of a list $\langle i_1,\ldots,i_j\rangle$ which can be
empty, and one index $i_k$. Since the task is to decide whether
$\varphi_{i_1}\wedge\ldots\wedge\varphi_{i_j}\Rightarrow\varphi_{i_k}$ the
line~\ref{algLine:ver1} needs to be altered to:\\ $\textsf{F} \leftarrow
\texttt{createConj}(\varphi_{i_1}, \ldots, \varphi_{i_j}, \neg
\varphi_{i_k})$\\ This is due to the fact discussed in the previous section,
namely that the satisfiability of $\varphi_{i_1} \wedge \ldots \wedge
\varphi_{i_j} \wedge \neg \varphi_{i_k}$ implies $\varphi_{i_1} \wedge \ldots
\wedge \varphi_{i_j} \not \Rightarrow \varphi_{i_k}$, i.e. $\varphi_{i_k}$ is
not redundant with respect to $\{\varphi_{i_1}, \ldots, \varphi_{i_j}\}$.

Discarding the B{\"u}chi automata in every iteration may seem unnecessarily
wasteful, especially since synchronous composition of two (and more) automata is
a feasible operation. However, the size of an automaton created by composition
is a multiplication of the sizes of the automata being composed. Furthermore, it
would not be possible to use the size optimising techniques employed in LTL to
B{\"u}chi translation. And these techniques work particularly well in our case,
because the translated formulae (conjunctions of requirements) have relatively
small nesting depth (maximal depth among requirements $+1$).

\begin{example}
  We will demonstrate the search for smallest inconsistent subsets on the
  following set of requirements.
  \begin{enumerate}
  \item A request to increase the cabin temperature $a$ may always be issued and
    each such request will eventually be granted by heating unit
    $b$. \hfill$\varphi_1=\LTLg(\LTLf a\wedge(a\Rightarrow\LTLf b))$
  \item There is only a finite amount of fuel for the heating
    unit. \hfill$\varphi_2=\LTLf\LTLg(\neg b)$
  \item The system must be robust to handle overriding command $c$ after which
    no request to increase the temperature may be
    issued. \hfill$\varphi_3=\LTLf(c\wedge c\Rightarrow\LTLg \neg a)$
  \item There will be no request to increase the temperature initially, but
    every heating will be preceded by a request which must be turned off after
    one time unit once the heating starts. Also the request is always issued
    until the heating starts.
    
    \textcolor{white}{a}\hfill$\varphi_4=\LTLg(a\LTLu b \wedge \LTLx
    b\Rightarrow(a\wedge\LTLx a\wedge\LTLx\LTLx\neg a))\wedge \neg a$
  \end{enumerate}

  Our consistency checking algorithm begins by checking in parallel the
  individual requirements $\varphi_1,\ldots,\varphi_4$ for consistency and also
  the whole set $\Gamma=\{\varphi_1,\ldots,\varphi_4\}$. It discovers that
  $\varphi_4$ is inconsistent on its own and that $\Gamma$ is also
  inconsistent. The fact that $\varphi_4$ is inconsistent is employed by not
  considering any of the sets containing $\varphi_4$,
  i.e. $\{\varphi_1,\varphi_4\}$, $\{\varphi_2,\varphi_4\}$,
  $\{\varphi_3,\varphi_4\}$, $\{\varphi_1,\varphi_2,\varphi_4\}$,
  $\{\varphi_1,\varphi_3,\varphi_4\}$, and
  $\{\varphi_2,\varphi_3,\varphi_4\}$. Next, the algorithm checks the supersets
  of the largest consistent sets, i.e. $\{\varphi_1,\varphi_2\}$,
  $\{\varphi_1,\varphi_3\}$, and $\{\varphi_2,\varphi_3\}$. Discovering that
  both $\{\varphi_1,\varphi_2\}$ and $\{\varphi_1,\varphi_3\}$ are inconsistent
  it terminates since $\{\varphi_1,\varphi_2,\varphi_3\}$ is necessarily also
  inconsistent. The smallest inconsistent subsets are thus $\{\varphi_4\}$,
  $\{\varphi_1,\varphi_2\}$, and $\{\varphi_1,\varphi_3\}$. \hfill$\triangle$
\end{example}

\section{Checking Completeness} \label{sec:compl}

The completeness checking is a little more involved: this is in fact the part
that provably cannot be fully automated. Hence the paper will first describe the
problem and then detail the semi-automatic solution proposed.

Let us assume that the user specifies three types of requirements: environmental
assumptions $\Gamma_{A}$, required behaviour $\Gamma_{R}$ and forbidden
behaviour $\Gamma_{F}$. The environmental assumptions represent the sensible
properties of the world a given system is to work in, e.g. ``\emph{The plane is
  either in the air or on the ground, but never in both these states at once}'',
in LTL this would read $\LTLg (ground\Leftrightarrow \neg air)$. The required
behaviour represents the functional requirements imposed on the system: the
system will not be correct if any of these is not satisfied, e.g. ``\emph{If the
  plane is ever to fly it must first be tested on the ground}'': $\LTLf air
\Rightarrow (ground\wedge tested)$. Dually, the forbidden behaviour contains
those patterns that the system must not display, e.g. ``\emph{Once tested the
  plane should be put into air within the next three time steps}'':
$\neg\ tested \LTLu (air\vee \LTLx air\vee \LTLx\LTLx air)$.

Assume henceforth the following simplifying notation for B{\"u}chi automata: let
$f$ be a propositional formula over capital Latin letters $A, B,\ldots$ barred
letters $\bar{A}, \bar{B},\ldots$ and Greek letters $\alpha, \beta,\ldots$,
where $A$ substitutes $\bigwedge_{\gamma\in \Gamma_{A}}\gamma$, $\bar{A}$ stands
for $\bigvee_{\gamma\in\Gamma_{A}}\gamma$ and all Greek letters represent simple
LTL formulae. Then $\mathcal{A}_f$ denotes such a B{\"u}chi automaton that
accepts all words satisfying the substituted $f$, e.g. $\mathcal{A}_{A\vee
  \bar{B} \wedge \varphi}$ accepts words satisfying $\bigwedge_{\gamma \in
  \Gamma_{A}}\gamma \vee \bigvee_{\gamma\in\Gamma_{B}}\gamma \wedge\varphi$. The
automaton $\mathcal{A}_A$ thus describes the part of the state space the user is
interested in and which the required and forbidden behaviour should together
submerge. That most commonly is not the case with freshly elicited requirements
and therefore the problem is the following: find sufficiently simple formulae
over $AP$ that would, together with formulae for $R$ and $F$, cover a large
portion of $\mathcal{A}_A$. In other words to find such $\varphi$ that
$\mathcal{A}_{R \vee \bar{F} \vee \varphi}$ covers as much of $\mathcal{A}_A$ as
possible.

In order to evaluate the size of the part of $\mathcal{A}_A$ covered by a single
formula, i.e. how \emph{much} of the possible behaviour is described by it, an
evaluation methodology for B{\"u}chi automata needs to be established. The plain
enumeration of all possible words accepted by an automaton is impractical given
the fact that B{\"u}chi automata operate over infinite words. Similarly, the
standard completeness metrics based on state coverage~\cite{CKV01,TK01} are
unsuitable because they do not allow for comparison of sets of formulae and they
require the underlying model. Equally inappropriate is to inspect only the
underlying directed graph because B{\"u}chi automata for different formulae may
have isomorphic underlying graphs.

The methodology proposed in this paper is based on the notion of
\emph{almost-simple} paths and \emph{directed partial coverage} function.

\begin{define} \label{def:almSim}
  Let $G$ be a directed graph. A \emph{path} $\pi$ in $G$ is a sequence of
  vertices $v_1, \ldots, v_n$ such that $\forall i: (v_i,v_{i+1})$ is an edge in
  $G$. A path is \emph{almost-simple} if no vertex appears on the path more than
  twice. The notion of almost-simplicity is also applicable to words in the case
  of B{\"u}chi automata.
\end{define}

With almost-simple paths one can enumerate the behavioural patterns of a
B{\"u}chi automaton without having to handle infinite paths. Clearly, it is a
heuristic approach and a considerable amount of information will be lost but
since all simple cycles will be considered (and thus all patterns of the
infinite behaviour) the resulting evaluation should provide sufficient
distinguishing capacity (as demonstrated in Section~\ref{sec:cas}).

Knowing which paths are interesting it is possible to propose a methodology that
would allow comparing two paths. There is, however, a difference between
B{\"u}chi automata that represent a computer system and those built using only
LTL formulae (ones that should restrict the behaviour of the system). The latter
automata use a different evaluation function $\hat{\nu}$ that assigns to every
edge a set of literals. The reason behind this is that the LTL-based automaton
only allows those edges (in the system automaton) for which their source vertex
has an evaluation compatible with the edge evaluation (now in the LTL
automaton).

\begin{define} \label{def:evalSim}
  Let $\mathcal{A}_1$ and $\mathcal{A}_2$ be two (LTL) B{\"u}chi automata over
  $AP$ and let $AP_L$ be the set of literals over $AP$. The \emph{directed
    partial coverage} function $\Lambda$ assigns to every pair of edge
  evaluations a rational number between 0 and 1, $\Lambda:2^{AP_L}\times$
  $2^{AP_L} \rightarrow\mathbb{Q}$. The evaluation works as follows
  \[\Lambda(A_1=\{l_{11},\ldots,l_{1n}\},A_2=\{l_{21},\ldots,l_{2m}\})=\left\{
    \begin{array}{ll}
      0 & \ \exists i,j: l_{1i}\equiv\neg l_{2j}\\
      p/m & \ \text{otherwise }\\
    \end{array} \right.
  \]
  where  $p=|A_1\cap A_2|$.
\end{define}

From this definition one can observe that $\Lambda$ is not symmetric. This is
intentional because the goal is to evaluate how much a path in $\mathcal{A}_2$
covers a path in $\mathcal{A}_1$. Hence the fact that there are some additional
restricting literals on an edge of $\mathcal{A}_1$ does not prevent automaton
$\mathcal{A}_2$ to display the required behaviour (the one observed in
$\mathcal{A}_1$).

The extension of coverage from edges to paths and to automata is based on
averaging over almost-simple paths. An almost-simple path $\pi_2$ of automaton
$\mathcal{A}_2$ covers an almost-simple path $\pi_1$ of automaton
$\mathcal{A}_1$ by $\Lambda(\pi_1,\pi_2)= \frac{\sum_{i=0}^{n}
  \Lambda(A_{1i},A_{2i})}{n}$ per cent, where $n$ is the number of edges and
$A_{ji}$ is the set of labels on $i$-th edge on $\pi_j$. Then automaton
$\mathcal{A}_2$ covers $\mathcal{A}_1$ by $\Lambda(\mathcal{A}_1,\mathcal{A}_2)
=\frac{\sum_{i=0}^{m} \max_{2i}\Lambda(\pi_{1i},\pi_{2i})}{m}$ per cent, where
$m$ is the number of almost-simple paths of $\mathcal{A}_1$ that end in an
accepting vertex. It follows that coverage of 100 per cent occurs when for every
almost-simple path of one automaton there is an almost-simple path in the other
automaton that exhibits similar behaviour.

\subsection{Implementation of Completeness Checking}

The high-level overview of the implementation of the proposed methodology is
based on partial coverage of almost-simple paths of $\mathcal{A}_A$. In other
words finding the most suitable path in $\mathcal{A}_{R \vee
  \bar{F}\vee\varphi}$ for every almost-simple path in $\mathcal{A}_A$, where
$\varphi$ is a sensible simple LTL formula that is proposed as a candidate for
completion. Finally, the suitability will be assessed as the average of partial
coverage over all edges on the path.

The output of such a procedure will be a sequence of candidate formulae, each
associated with an estimated coverage (a number between 0 and 1) the addition of
this particular formula would entail. The candidate formulae are added in a
sequence so that the best unused formula is selected in every round. Finally,
the coverage is always related to $\mathcal{A}_A$ and, thus, if some behaviour
that cannot be observed in $\mathcal{A}_A$ is added with a candidate formula
this addition will neither improve nor degrade the coverage.

\begin{example} \label{exa:eval}
  The method of B{\"u}chi automata evaluation will be partially exemplified
  using Figure~\ref{fig:exampleEval}. The example only shows the enumeration of
  almost-simple paths and the partial coverage of two paths. What remains to the
  complete methodology will be shown more structurally in
  Algorithm~\ref{alg:compl}. The enumeration of almost-simple paths of
  $\mathcal{A}_a$ in Figure~$1a)$ should be straightforward, part the fact that
  a path is represented as a sequence of edges for simplicity. Let us assume
  that $\mathcal{A}_b$ is the original automaton and $\mathcal{A}_c$ is being
  evaluated for how thoroughly it covers $\mathcal{A}_b$. There are 4
  almost-simple paths in $\mathcal{A}_b$, one of them is $\pi=\langle a;\neg
  a,b;c,a;a\rangle$. The partial coverage between the first edge of $\pi$ and
  the first edge of $\mathcal{A}_c$ (there is only one possibility) is $0.5$,
  since there is the excessive literal $d$. The coverage between the second
  edges is also $0.5$, but only because of $\neg c$ in $\mathcal{A}_c$; the
  superfluous literal $\neg a$ restricts only the behaviour of
  $\mathcal{A}_b$. Finally, the average similarity between $\pi$ and the
  respective path in $\mathcal{A}_c$ is $0.75$ and it is approximately $0.7$
  between the two automata.\hfill$\triangle$
\end{example}

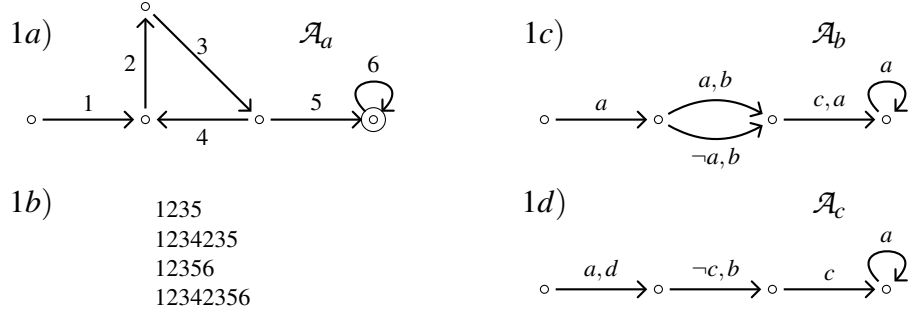
\begin{figure}[t]
  \centering
  \begin{tikzpicture}[scale=1.5,auto,swap]
  \clip (-0.2,-1.7) rectangle (7.7,1.2);
  \foreach \pos/\name in {{(0,0)/a}, {(1,0)/b}, {(1,1)/c}, {(2,0)/d}} {
    \node[vertex] (\name) at \pos {};
    \draw \pos circle (1pt);
  }

  \foreach \pos/\name in {{(4.5,0)/f}, {(5.5,0)/g}, {(6.5,0)/h}, {(7.5,0)/i}} {
    \node[vertex] (\name) at \pos {};
    \draw \pos circle (1pt);
  }

  \foreach \pos/\name in {{(4.5,-1.5)/j}, {(5.5,-1.5)/k}, {(6.5,-1.5)/l},
    {(7.5,-1.5)/m}} {
    \node[vertex] (\name) at \pos {};
    \draw \pos circle (1pt);
  }

  \foreach \pos/\name in {{(3,0)/e}} {
    \node[vertex] (\name) at \pos {};
    \draw \pos circle (1pt);
    \draw \pos circle (3pt);
  }

  \node[vertex] (A) at (0,0.75) {\large $1a)$};
  \node[vertex] (A) at (0,-0.75) {\large $1b)$};
  \node[vertex] (A) at (4.5,0.75) {\large $1c)$};
  \node[vertex] (A) at (4.5,-0.75) {\large $1d)$};

  \node[vertex] (B) at (1.5,-1.2)
       {$\begin{array}{l}
           1235 \\
           1234235 \\
           12356 \\
           12342356
         \end{array}
        $};

  \node[vertex] (C) at (2.5,0.75) {\large $\mathcal{A}_a$};
  \node[vertex] (C) at (7,0.75) {\large $\mathcal{A}_b$};
  \node[vertex] (C) at (7,-0.75) {\large $\mathcal{A}_c$};

  \begin{scope}[>=angle 90]
    \path[edge] (a) -- node [above] {1} (b);
    \path[edge] (b) -- node [left] {2} (c);
    \path[edge] (c) -- node [above] {3} (d);
    \path[edge] (d) -- node [below] {4} (b);
    \path[edge] (d) -- node [above] {5} (e);
    \path[edge] (e) to [in=45,out=135,loop] node [above] {6} (e);

    \path[edge] (f) -- node [above] {$a$} (g);
    \path[edge] (g) to [bend left=30] node [above] {$a,b$} (h);
    \path[edge] (g) to [bend right=30] node [below] {$\neg a,b$} (h);
    \path[edge] (h) -- node [above] {$c,a$} (i);
    \path[edge] (i) to [in=45,out=135,loop] node [above] {$a$} (i);

    \path[edge] (j) -- node [above] {$a,d$} (k);
    \path[edge] (k) -- node [above] {$\neg c,b$} (l);
    \path[edge] (l) -- node [above] {$c$} (m);
    \path[edge] (m) to [in=45,out=135,loop] node [above] {$a$} (m);

  \end{scope}
\end{tikzpicture}
  \caption{$1a)$ Example B{\"u}chi automaton $\mathcal{A}_a$; $1b)$ All
    almost-simple paths of $\mathcal{A}_a$; $1c)$ and $1d)$ are two different
    B{\"u}chi automata with relatively similar evaluations of almost-simple
    paths (see Example~\ref{exa:eval}).}
  \label{fig:exampleEval}
\end{figure}

The topmost level of the completeness evaluation methodology is shown as
Algorithm~\ref{alg:compl}. As input this function requires the three sets of
user defined requirements, the set of candidate formulae and the number of
formulae the algorithm needs to select. On lines~\ref{algLine:c1}
and~\ref{algLine:c2} the formulae for conjunction of assumptions and user
requirements (both required and forbidden) are created. They will be used later
to form larger formulae to be translated into B{\"u}chi automata and evaluated
for completeness but, for now, they need to be kept separate. Next step is to
enumerate the almost-simple paths of $\mathcal{A}_A$ for later comparison,
i.e. a baseline state space that the formulae from $\Gamma_{Cand}$ should cover.

\begin{algorithm}[t]
  {\normalsize
  \DontPrintSemicolon
\SetKwData{Aut}{A}\SetKwData{EvalA}{pathsBA}\SetKwData{Min}{max}
\SetKwData{Cur}{cur}
\SetKwFunction{eBA}{enumeratePaths}\SetKwFunction{tBA}{transform2BA}
\SetKwFunction{Print}{print}\SetKwFunction{aPC}{avrPathCov}
\SetKwInOut{Input}{Input}\SetKwInOut{Output}{Output}
\BlankLine
\Input{$\Gamma_{A}$, $\Gamma_{R}$, $\Gamma_{F}$, $\Gamma_{Cand}$, $n$}
\Output{Best coverage for $1\ldots n$ formulae from $\Gamma_{Cand}$}
\BlankLine
$\gamma_{Assum}\leftarrow\bigwedge_{\gamma\in\Gamma_{A}}\gamma$\label{algLine:c1}\;
$\gamma_{Desc}\leftarrow\bigwedge_{\gamma\in\Gamma_{R}}\gamma\vee\bigvee_{\gamma\in\Gamma_{F}}\gamma$
\label{algLine:c2}\;
$\Aut\leftarrow\tBA(\gamma_{Assum})$\;
$\EvalA\leftarrow\eBA(\Aut)$\;
\For{$i=1\ldots n$}{
  $\Min\leftarrow \infty$\;
  \ForEach{$\gamma\in\Gamma_{Cand}$}{
    $\gamma_{Test}\leftarrow\gamma_{Desc}\vee\gamma$
      \label{algLine:c3}\;
    $\Aut\leftarrow\tBA(\gamma_{Test})$\;
    $\Cur\leftarrow\aPC(\Aut,\EvalA)$\;
    \If{$\Min<\Cur$}{
      $\Min\leftarrow\Cur$\;
      $\gamma_{Max}\leftarrow\gamma$\;
    }
  }
  $\Print($ ``Best coverage in $i$-th round is $\Min$.'' $)$\;
  $\gamma_{Desc}\leftarrow\gamma_{Desc}\vee\gamma_{Max}$\;
  $\Gamma_{Cand}\leftarrow\Gamma{Cand}\setminus\{\gamma_{Max}\}$\;
}
\BlankLine

  \caption{Completeness Evaluation}
  \label{alg:compl}}
\end{algorithm}

The rest of the algorithm forms a cycle that iteratively evaluates all
candidates from $\Gamma_{Cand}$ (see line~\ref{algLine:c3} where the
corresponding formula is being formed). Among the candidate formulae the one
with the best coverage of the paths is selected and subsequently added to
the covering system.

Functions \texttt{enumeratePaths} and \texttt{avrPathCov} are similar extensions
of the BFS algorithms. Unlike BFS, however, they do not keep the set of visited
vertices to allow state revisiting (twice in case of \texttt{enumeratePaths} and
arbitrary number of times in case of \texttt{avrPathCov}). The
\texttt{avrPathCov} search is executed once for every path it receives as input
and stops after inspecting all paths to the length of the input path or if the
current search path is incompatible (see Definition~\ref{def:evalSim}).


\section{Experimental Evaluation}

All three sanity checking algorithms were implemented as an extension of the
parallel explicit-state LTL model checker DiVinE~\cite{BBC10}. From the many
facilities offered by this tool, only the LTL to B{\"u}chi translation was
used. Similarly as the original tool also this extension to sanity checking was
implemented using parallel computation.

\subsection{Experiments with Random Formulae}

The first set of experiments was conducted on randomly generated LTL
formulae. The motivation behind using random formulae is to demonstrate the
reduction in the number of consistency checks. A naive algorithm for detecting
inconsistent/redundant subsets requires an exponential number of such
checks. Generating a large number of requirement collections with varying
relations among individual requirements allows a more representative results
than a few, real-world collections.  

In order for the experiments to be as realistic as possible (avoiding trivial or
exceedingly complex formulae) requirements with various nesting depths were
generated. Nesting depth denotes the depth of the syntactic tree of a
formula. Statistics about the most common formulae show, e.g. in~\cite{DAC98},
that the nesting is rarely higher than 5 and is 3 on average. Following these
observations, the generating algorithm takes as input the desired number $n$ of
formulae and produces: $n/10$ formulae of nesting 5, $9n/60$ of nesting 1, $n/6$
of nesting 4, $n/4$ of nesting 2 and $n/3$ of nesting 3. Finally, the number of
atomic propositions is also chosen according to $n$ (it is $n/3$) so that the
formulae would all contribute to the same state space.

All experiments were run on a dedicated Linux workstation with quad core Intel
Xeon 5130 @ 2GHz and 16GB RAM. The codes were compiled with optimisation options
\texttt{-O2} using GCC version 4.3.2. Since the running times and even the
number of checks needed for completion of all proposed algorithms differ for
every set of formulae, the experiments were ran multiple times. The sensible
number of formulae starts at 8: for less formulae the running time is
negligible. Experimental tests for consistency and redundancy were executed for up
to 15 formulae and for each number the experiment was repeated 25 times.

Figure~\ref{fig:conTime} summarises the running times for consistency checking
experiments. For every set of experiments (on the same number of formulae) there
is one box capturing median, extremes and the quartiles for that set of
experiments. From the figure it is clear that despite the optimisation
techniques employed in the algorithm both median and maximal running times
increase exponentially with the number of formulae. On the other hand there are
some cases for which presented optimisations prevented the exponential blow-up
as is observable from the minimal running times.

\begin{figure}[t]
  \centering \includegraphics[width=0.8\textwidth]{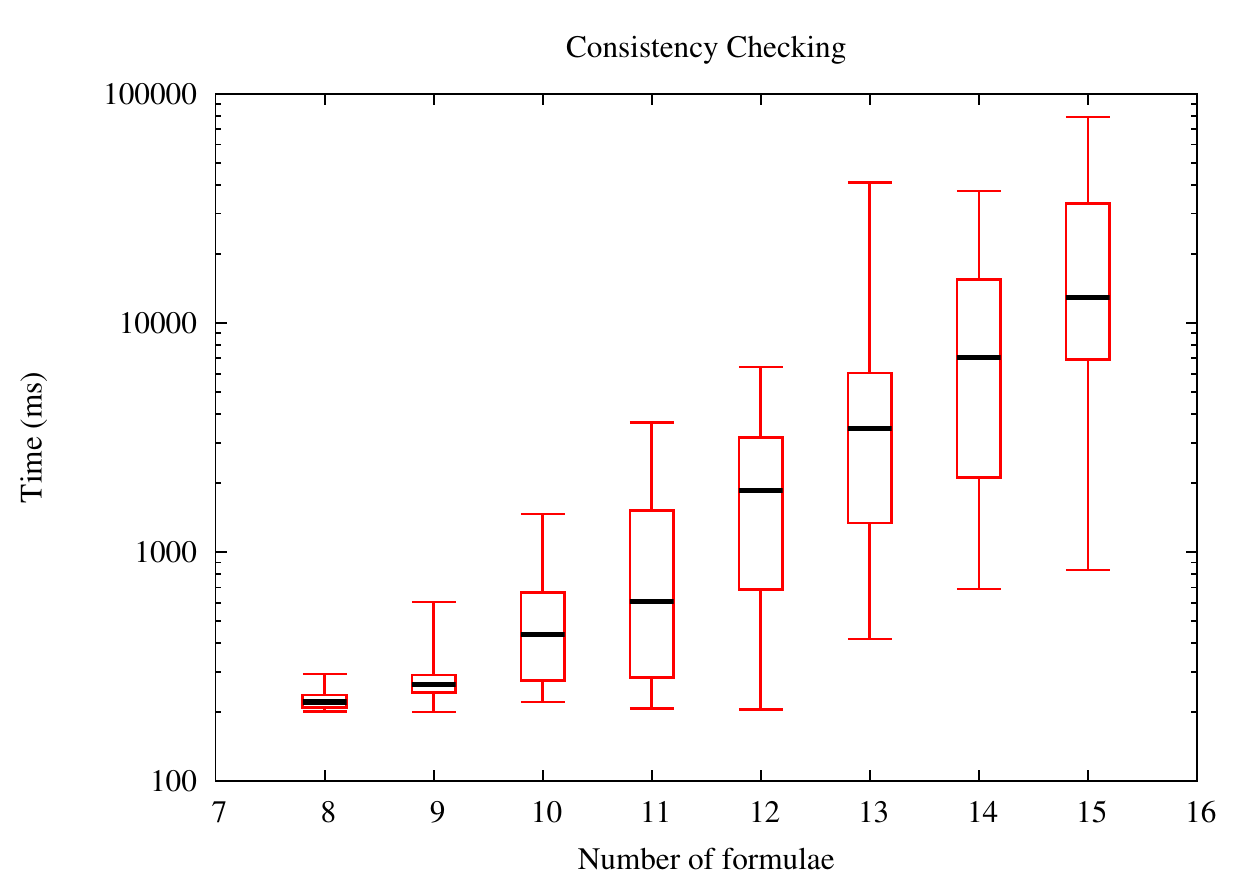}
  \caption{Log-plot summarising the time complexity of consistency
    checking.}
  \label{fig:conTime}
\end{figure}

Figure~\ref{fig:vacChecks} illustrates the discrepancy between the number of
combinations of formulae and the number of redundancy checks that were actually
performed. The number of combinations for $n$ formulae is $n*2^{n-1}$ but the
optimisation often led to much smaller number. As one can see from the
experiments on 9 formulae, it is potentially necessary to check almost all the
combinations but the method proposed in this paper requires on average less than
10 per cent of the checks and the relative number decreases with the number of
formulae.

The actual running times, as reported in Figure~\ref{fig:conTime}, depend
crucially on the chosen method for consistency checking: automata-based
Algorithm~\ref{alg:ver} in our case. Our algorithm for enumerating all smallest
inconsistent subsets is robust with respect to the consistency checking
algorithm used. Hence any other consistency (or realisability) tool could be
used, but that would only affect the actual running times, but not the number of
checks, reported in Figure~\ref{fig:vacChecks}.

\begin{figure}[t]
  \centering \includegraphics[width=0.8\textwidth]{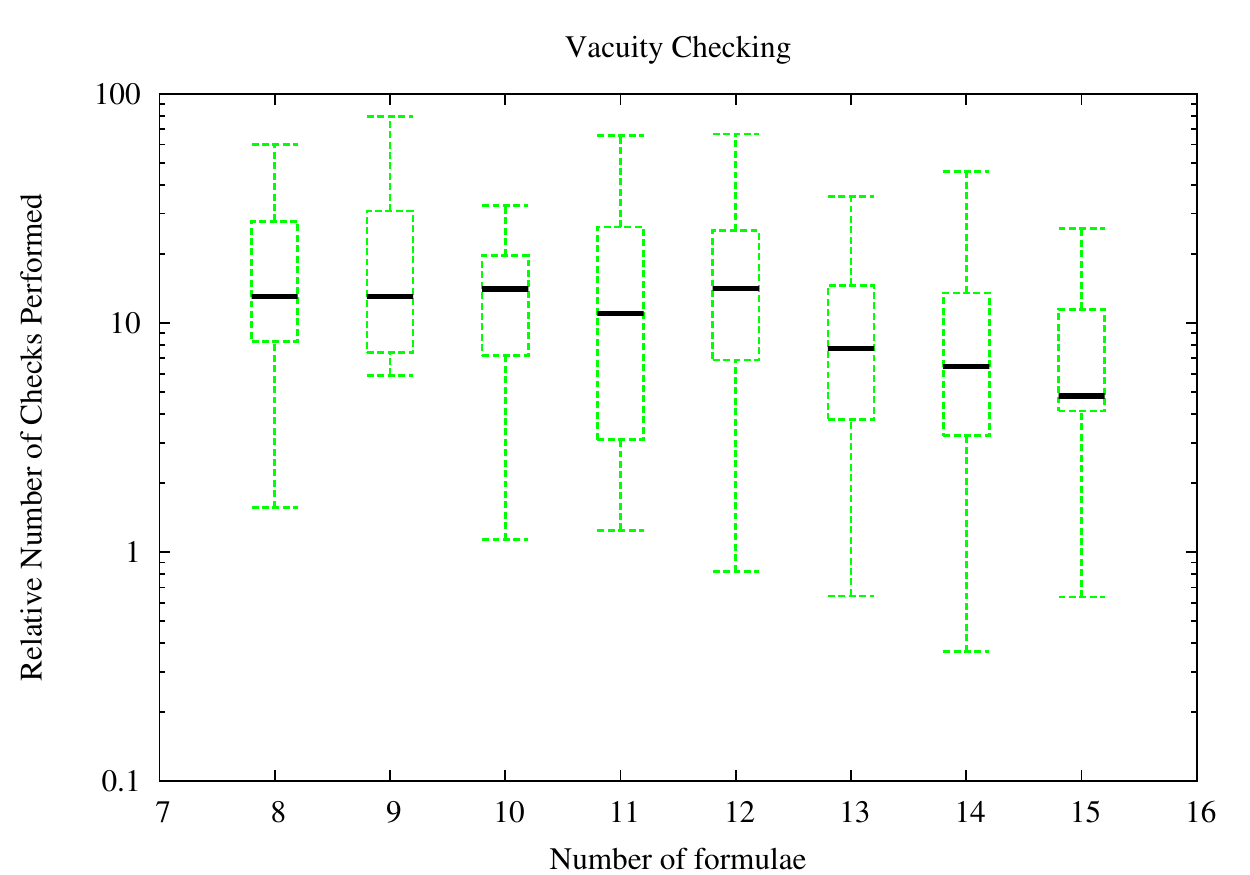}
  \caption{Log-plot with the relative number of checks for redundancy
    checking.}
  \label{fig:vacChecks}
\end{figure}

\subsection{Case Study: Aeroplane Control System} \label{sec:cas}

In order to demonstrate the full potential of the proposed sanity checking, we
apply all the steps in a sequence. First, the whole requirements document is
checked for consistency and redundancy and only the consistent and non-redundant
subset is then evaluated with respect to its completeness. A sensible exposition
of the effectiveness of completeness evaluation proposed in this paper requires
more elaborate approach than using random formulae as the input requirements
document.

For that purpose a case study has been devised that demonstrates the capacity to
assess coverage of requirements and to recommend suitable coverage-improving LTL
formulae. The requirements document was written by a requirements engineer whose
task was to evaluate to method, and we only use random formulae as candidates
for coverage improvement. The candidate formulae were built based on the atomic
proposition that appeared in the input requirements and only very simple
generated formulae were selected. It was also required that the candidates do
not form an inconsistent or tautological set. Alternatively,
pattern-based~\cite{DAC98} formulae could be used as candidates. The methodology
is general enough to allow semi-random candidates generated using patterns from
input formulae. For example if an implication is used as the input formula, the
antecedent may not be required by the other formulae which may not be according
to user's expectations.

The case study attempts to propose a system of LTL formulae that should control
the flight and more specifically the landing of an aeroplane. The LTL formulae
and the atomic propositions they use are summarised in
Figure~\ref{fig:case}. The requirements are divided into 3 categories similarly
as in the text: $A$ requirements represent assumptions and $R$ and $F$ stand for
required and forbidden behaviour, respectively. The grayed formulae were found
either redundant or inconsistent and were not included in the completeness
evaluation.

\begin{itemize}
\item[A1] The plane is on the ground only if its speed is less or equal to 200
  mph.
\item[A2] The plane will eventually land and every landing will result in the
  plane being on the ground.
\item[A3] If the plane is not landing its speed is above 200 mph.
\item[A4] Whenever the undercarriage is down, first the speed must be at most
  100 mph and second the plane must eventually land on the ground.
\item[A5] \textcolor{gray}{Whenever the plane is not landing its height above
  ground is not zero.}
\item[A6] \textcolor{gray}{At any time during the flight it holds that either
  the plane if flying faster than 200 mph or it eventually touches the ground.}
\vspace{1em}
\item[R1] The landing process entails that eventually: (1) the speed will
  remain below 200 mph indefinitely and (2) after some finite time the speed
  will also remain below 100 mph.
\item[R2] The landing process also entails that the plane should eventually slow
  down from 200 mph to 100 mph (during which the speed never goes above 200
  mph).
\item[R3] Whenever flying faster than 200 mph, the undercarriage must be
  retracted.
\item[R4] At some point in the future, the plane should commence landing, then
  detract the undercarriage until finally the speed is below 100 mph.
\item[R5] \textcolor{gray}{In the future the plane will be decommissioned: it
  will land and then either never flying faster than 100 mph or always slowing
  down below 200 mph.}
\item[R6] \textcolor{gray}{At some point in the future, the speed will be below
  200 mph and the undercarriage retracted. Immediately before that and while the
  speed is still above 200 mph the undercarriage is detracted in two consecutive
  steps.}
\vspace{1em}
\item[F1] The plane will eventually be on the ground after which it will never
  take off again.
\end{itemize}

\begin{figure}[t]
  \centering
  \begin{tikzpicture}[scale=1.5,auto,swap]
  \clip (-1.1,-1.1) rectangle (7.5,1.2);
  \node[vertex] (A) at (0,1) {\textbf{Atomic Propositions}};
  \node[vertex] (B) at (0,0)
    {$\begin{array}{rcl}
        a & \equiv & [\text{height}=0] \\
        b & \equiv & [\text{speed}\leq 200] \\
        l & \equiv & [\text{landing}] \\
        u & \equiv & [\text{undercarriage}] \\
        c & \equiv & [\text{speed}\leq 100]
      \end{array}$};

    \node[vertex] (C) at (4.2,1) {\textbf{LTL Requirements}};
  \node[vertex] (D) at (4.2,0)
    {$\begin{array}{lcllcl}
        A1 & : & \LTLg(a\Leftrightarrow b) &\ 
        R1 & : & \LTLg(l\Rightarrow \LTLf(\LTLg b\wedge \LTLf\LTLg c))\\

        A2 & : & \LTLf l\wedge \LTLg (l\Rightarrow \LTLf a) &\ 
        R2 & : & \LTLg(l\Rightarrow \LTLf(b \LTLu c))\\

        A3 & : & \LTLg(\neg l\Rightarrow\neg b) &\ 
        R3 & : & \LTLg(\neg b\Rightarrow\neg u)\\

        A4 & : & \LTLg(u\Rightarrow \LTLf a\wedge u\Rightarrow c) &\ 
        R4 & : & \LTLf(l \LTLu (u\LTLu c))\\

        A5' & : & \LTLg(\neg l\Rightarrow\neg a) &\
        R5' & : & \LTLf(l\wedge\LTLg(\neg c\vee\LTLf(\neg b)))\\

        A6' & : & \LTLg(\neg b\vee\LTLf\ a) &\
        R6' & : & \LTLf(\neg b\wedge\neg u\wedge\LTLx(\neg b\wedge u)\wedge\LTLx\LTLx(b\wedge u))\\

      \end{array}$};

    \node[vertex] (E) at (4,-1)
      {$\begin{array}{rcl}
          F1 & : & \LTLf(a\wedge \LTLg(\neg a))
        \end{array}$};
\end{tikzpicture}
  \caption{The two tables explain the shorthands for atomic propositions and
    list the LTL requirements.}
  \label{fig:case}
\end{figure}

\paragraph{Consistency}

The three categories of requirements are checked for consistency and redundancy
separately. The environmental assumptions form a consistent set. The set of
required behaviour contained two smallest inconsistent subsets. The requirement
R4 is inconsistent with R6, hence we have choose R4 as more important and
exclude R6 from the collection. The second inconsistent subset was: R1, R2,
R5. There we chose to preserve R1 and R2, leaving out the possibility to
decommission a plane.

\paragraph{Redundancy}

After regaining consistency of the required behaviour, our sanity checking
procedure detects no redundancy in this set. On the other hand there are two
redundant environmental assumptions. The assumption A5 is implied by the
conjunction of A1 and A3 and is thus redundant. It also holds that any
environment in which A2 and A3 are true also satisfies A6. Hence both A5 and A6
can be removed without changing the set of sensible behaviour in the
environment.

\paragraph{Completeness}

Initially, the remaining required and forbidden behaviour together does not
cover the environment assumptions at all, i.e. the coverage is 0. There simply
are so many possible paths allowed by the assumptions which are not considered
in the requirements, that our metric evaluated the coverage to 0. Given the way
the coverage is calculated, it must be the case that for every path in the
assumption automaton there was not a single path in the requirements automaton
that would not violate some of the edge labels of the assumptions. There is thus
a considerable space for improvement. The completeness improving process we are
about to describe begins by generating a set of candidate requirements among
which those with the best resulting coverage are selected. Table~\ref{tab:iter}
summarises the process.

The first formula selected by the Algorithm~\ref{alg:compl} and leading to
coverage of $9.1$ per cent was a simple $\varphi_1=\LTLg(\neg l)$: ``\emph{The
  plane will never land}''. Not particularly interesting per se, nonetheless
emphasising the fact that without this requirement, landing was never
required. The formula $\varphi_1$ should thus be included among the forbidden
behaviour. Unlike $\varphi_1$, the formula generated as second, $\LTLf(\neg
b\wedge\neg l)$ : ``\emph{At some point the plane will fly faster than 200 mph
  and will not be in the process of landing}'', should be included among the
required behaviour. This formula points out that the required behaviour only
specifies what should happen after landing, unlike the assumptions, which also
require flight. The third and final formula was $\LTLf(\neg l \LTLu (a\vee b))$:
``\emph{Eventually it will hold that the landing does not commence until either:
  (1) the plane is already on the ground or (2) the speed decreased below 200
  mph}''. This last formula connects flight and landing and its addition (among
the required behaviour) entails coverage of $54.9$ per cent.

There are two conclusion one can draw from this case study with respect to the
completeness part of sanity checking. First, the requirements are generated
automatically, they are relevant and already formalised, thus amenable to
further testing and verification. Second, the new requirements can improve the
insight of the engineer into the problem. Both of these properties are valuable
on their own, even though the method did not finish with 100 per cent
coverage. The method produced further, coverage-improving formulae but even the
best formulae only negligible improved the coverage and thus the engineer
decided to stop the process.

\begin{table}[t]
\caption{Iterations of the completeness evaluation algorithm for the Aeroplane
  Control System.}
\label{tab:iter}
\centering
\begin{tabular}{|l|l|l|l|}
\hline
iteration & suggested formula & resulting system & coverage of $\mathcal{A}_{A}$ by
$\mathcal{A}_n$\\ 
\hline
1 & $\varphi_1=\LTLg(\neg l)$ &
$\mathcal{A}_1=\mathcal{A}_{R\vee\overline{F}\vee\overline{\varphi_1}}$ &
9.1\%\\ 
\hline
2 & $\varphi_2=\LTLf(\neg b\wedge\neg l)$ &
$\mathcal{A}_2=\mathcal{A}_{R\vee\overline{F}\vee\overline{\varphi_1}\vee\varphi_2}$ & 39.4\%\\ 
\hline
3 & $\varphi_3=\LTLf(\neg l\LTLu(a\vee b))$ &
$\mathcal{A}_3=\mathcal{A}_{R\vee\overline{F}\vee\overline{\varphi_1}\vee\varphi_2\vee\varphi_3}$
& 54.9\%\\
\hline
\end{tabular}
\end{table}

\subsection{Alternative LTL to B{\"u}chi Translations}

Although the running time experiments were one of the priorities of the original
paper~\cite{BBB12}, their purpose was to investigate the effects of the proposed
optimisations and not the concrete time measurements. In other words, the
subject of our investigation was the asymptotic behaviour of the algorithms and
how effectively does the algorithm scale with the number of formulae in the
average and in the extrema. Given our recent cooperation with Honeywell and
their interest in checking sanity of real-world, industry-level sets of
requirements, the actual running times became of particular interest as well.

Initial experiments revealed, however, that on larger sets of requirements and
especially when more complex individual requirements were involved, the proposed
sanity checking algorithms were unable to cope with the computational
complexity. The bottleneck proved to be the algorithm translating LTL formulae
-- conjunctions of the original requirements -- into B{\"u}chi automata. The
translator incorporated in DiVinE is sufficiently fast on small formulae that
commonly occur in software verification and was never intended to be used with
larger conjunctions. Thus, to enable checking sanity of real-world sets of
requirements, we have reimplemented the algorithms to employ the
state-of-the-art LTL to B{\"u}chi translator SPOT~\cite{Dur11} in a tool called
Looney\footnote{As we are checking the sanity of requirements, the translator
  helps us to \emph{spot} the \emph{looney} ones. The tool is available at
  \url{http://anna.fi.muni.cz/~xbauch/code.html\#sanity}.}. As the case study
summarising our cooperation with Honeywell in Section~\ref{sec:eval} shows,
Looney is able to process much larger formulae than the original tool and also
to incorporate a larger number of formulae into the sanity checking process.

The reimplementation of consistency was relatively straightforward since SPOT
can be interfaced via a shared library with an arbitrary C++ application. Even
though this enforced us to use the internal implementation of formulae from
SPOT, which was different the one we had before, SPOT offers an adequate set of
formulae-manipulating operations that allowed the translation of our algorithms,
without their needing to be considerably modified. The reimplementation of the
coverage algorithms was complicated by the fact that the product of SPOT is a
\emph{transition-based} B{\"u}chi automaton, whereas the algorithm of
Section~\ref{sec:compl} expects \emph{state-based} automata: the difference is
that transitions rather than states are denoted as accepting.

This shift requires modification of the coverage evaluation, more precisely the
definition of $\Lambda(\mathcal{A}_1,\mathcal{A}_2)$ of the coverage among
automata. Before, the number $m$ in the denominator referred to the number of
almost-simple paths of $\mathcal{A}_1$ that ended in an accepting vertex
(state). We propose to redefine $m$ as the number of almost-simple paths that
end in an accepting edge (transition). It is clear that the coverage reported by
the original and the new metrics would differ: first, because the automata are
generated using different methods and may thus be structurally different and,
second, because the redefined metric considers different sets of almost-simple
paths. With the alternative definition of $m$ one can easily incorporate
translators that produce transition-based automata (though parsing these
automata may still require some degree of modification due to incompatibility of
internal structures representing the automata).

Therefore, in order to reestablish admissibility of the new evaluation as an
adequate coverage metric, we have repeated the experiment with the aeroplane
control system. As was to be expected given the above argumentation, the
concrete evaluation of individual formulae differed from our original
experiments. Specifically, the formula first generated ($\LTLg(\neg l)$) was
evaluated with coverage 12.9 per cent (instead of 9.1); the second formula
$\LTLf(\neg b\wedge\neg l)$ with 35.0 (instead of 39.4); and the third formula
$\LTLf(\neg l \LTLu (a\vee b))$ with 61.1 (instead of 54.9). Hence there are two
observation to be made. First, even though the coverage numbers differ for
individual formulae, the same triple was produced by the modified
algorithm. Moreover, the three formulae were produced in the same order, thus
the new metric behaves appropriately even when a new set of requirements is
obtained as a disjunction of the old set with the generated coverage-improving
formula.

\subsection{Industrial Evaluation} \label{sec:eval}

As a~further extension to the original paper~\cite{BBB12}, we have evaluated our
consistency and redundancy checking method on a~real-life sets of requirements
obtained in cooperation with Honeywell International. We did not extend this
part of evaluation to the completeness checking since that would require writing
a new set of requirements with separating the assumptions. The case study
consisted of four Simulink models; each has been associated with a~set of
formalised requirements in the form of LTL formulae. The precise requirements
are unfortunately confidential, but we provide statistics regarding the
complexity of the requirement documents. The formalisation of the requirements
that were originally given in natural language form has been done with the help
of the ForReq tool. The tool allows the user to write formal requirements with
the help of the specification patterns~\cite{DAC98} and some of their
extensions. See~\cite{BBBKR12} for a~detailed description of the tool.

Out of the four models, one (Lift) was a~toy model used in Honeywell to evaluate
the capabilities of ForReq, three (VoterCore, AGS, pbit) were models of
real-world avionics subsystems. As the focus of this paper is model-less sanity
checking, we ignored the models themselves and only considered the formalised
requirements. The results of our experiments are summarised in
Table~\ref{tbl:hw}.

The table reports that the largest collection of requirements consisted of 10
formulae, which can hardly be considered as a demonstration of scaling to
industrial level. Even though some of the requirements were extremely
complicated (with as many as 100 LTL operators), there is clearly a large space
for improvement. It would certainly be possible to use a different algorithm for
performing a single consistency check, but the number of these checks that needs
to be performed is an independent issue left for future work. In all four cases,
the whole set of requirements was consistent. As for redundancy, we have
identified two cases where a~formula was implied by another one. This is shown
in the redundancy result column. The number displayed in the mem(MB) column
represent the maximal (peak) amount of memory used.

The experiments, with the exception of the last one, were run on a dedicated
Linux workstation with Intel Xeon E5420 @ 2.5 GHz and 8 GB RAM. The last
experiment (pbit) was run on a Linux server with Intel X7560 @ 2.27 GHz and 450
GB RAM. The last experiment was a~consistency check only as the redundancy
checking took too much time. The reason for such great consumption of time 
(and memory) is that the formulae in the pbit set of requirements are rather
complex, some of them even using precise specification of time. For more
details about the time extension of LTL and its translation to standard LTL,
see~\cite{BBBKR12}.

\begin{table}
\caption{Results of consistency and redundancy checking on an
industrial case study.}
\label{tbl:hw}
\centering
\begin{tabular}{|c|c|c|r|r|c|r|r|}
\hline
model & no. of & \multicolumn{3}{c|}{consistency} & 
\multicolumn{3}{c|}{redundancy}  \\
name & \ formulae\  & 
\multicolumn{1}{c}{\ result\ } & 
\multicolumn{1}{c}{\ time(s)\ } & 
\multicolumn{1}{c|}{\ mem(MB)\ } & 
\multicolumn{1}{c}{\ result\ } & 
\multicolumn{1}{c}{\ time(s)\ } & 
\multicolumn{1}{c|}{\ mem(MB)\ }
\\\hline\hline
\ VoterCore\  & 3 & OK & 0.1 & 21.6 & OK & 0.4 & 22.4 \\\hline
AGS & 5 & OK & 0.1 & 23.5 & \ $\varphi_4 \Rightarrow \varphi_2$\  & 0.7 & 23.6 \\
\hline
Lift & 10 & OK & 73.6 & 58.5 & \ $\varphi_1 \Rightarrow \varphi_5$\  & \ \ 9460.5 & \ \ 1474.5 \\\hline
pbit & 10 & OK & \ \ 210879.0 & 22\,649.0 & --- & 
\multicolumn{1}{|c|}{---} & 
\multicolumn{1}{|c|}{---} \\\hline
\end{tabular}

\end{table}

\section{Conclusion}

This paper further expands the incorporation of formal methods into software
development. Aiming specifically at the requirements stage we propose a novel
approach to sanity checking of requirements formalised in LTL formulae. Our
approach is comprehensive in scope integrating consistency, redundancy and
completeness checking to allow the user (or a requirements engineer) to produce
a high quality set of requirements easier. The expert knowledge of LTL is not a
prerequisite for using the method, given the existence of automated translator
from English. On the other hand, there are certain limitations with respect to
the number of requirements checked concurrently, and larger collections may
require manually dividing into smaller sets.

Especially in the earliest stages of the requirements elicitation, when the
engineers have to work with highly incomplete, and often contradictory sets of
requirements, could our sanity checking framework be of great use. Other
realisability checking tools can also (and with better running times) be used on
the nearly final set of requirements but our process uniquely target these
crucial early states. First, finding all inconsistent subsets instead of merely
one helps to more accurately pinpoint the source of inconsistency, which may not
have been elicited yet. Second, once working with a consistent set, the
requirements engineer gets further feedback from the redundancy analysis, with
every minimal redundancy witness pointing to a potential lack of understanding
of the system being designed. Finally, the semi-interactive coverage-improving
generation of new requirements attempts to partly automate the process of making
the set of requirements as detailed as to prevent later development stages from
misinterpreting a less thoroughly covered aspect of the system. We demonstrate
this potential application on case studies of avionics systems.

The above described addition to the standard elicitation process may seem
intrusive, but the feedback we have gather from requirements engineers at
Honeywell were mostly positive. Especially given the tool support for each
individual steps of the process -- pattern-based automatic translation from
natural language requirements to LTL and the subsequent fully- and
semi-automated techniques for sanity checking -- a fast adoption of the
techniques was observed. Many of the engineers reported a non-trivial initial
investment to assume the basics of LTL, but in most cases the resulting overall
savings in effort outweighed the initial adoption effort. After the first week,
the overhead of tool-supported formalisation into LTL was less than 5 minutes
per requirement, while the average time needed for corresponding stages of
requirements elicitation decreased by 18 per cent.

We have also accumulated a considerable amount of experience from employing the
sanity checking on real-world sets of requirements during our cooperation with
Honeywell International. Of particular importance was the realisation that the
time requirements of our original implementation effectively prevented the use
of sanity checking on industrial scale. In order to ameliorate the poor timing
results we have reimplemented the sanity checking algorithms using the
state-of-the-art LTL to B{\"u}chi translator SPOT, which accelerated the process
by several orders on magnitude on larger sets of formulae. Another important
observation was that the concept of model-based vacuity is desirable to be
translated into the model-free environment more directly than as redundancy. We
thus extend the sanity checking process with the capacity to generate vacuity
witnesses (using the formally described theory of path-quantified
formulae). Finally, we summarise these and other observations in a case
consisting of four sets of requirements gathered during the development of
industry-scale systems.

One direction of future research is the pattern-based candidate selection
mentioned above. Even though the selected candidates were relatively sensible in
presented experiments, using random formulae can produce useless
results. Finally, experimental evaluation on real-life requirements and
subsequent incorporation into a toolchain facilitating model-based development
are the long term goals of the presented research. This paper also lacks formal
definition of \emph{total coverage} (which the proposed partial coverage merely
approximates). We intend to formulate an appropriate definition using uniform
probability distribution: which would also allow to compute total coverage
without approximation and would not be biased by concrete LTL to BA
translation. That solution, however, is not very practical since the underlying
automata translation is doubly exponential.

\bibliographystyle{plain}
\bibliography{faoc-sefm}

\end{document}